\RequirePackage{fix-cm}
\documentclass[smallextended]{svjour3}       
\smartqed  
\usepackage{lineno,hyperref}
\usepackage{amsmath}
\usepackage{makecell}
\usepackage{algorithmicx}
\usepackage{hyperref}
\usepackage[export]{adjustbox}
\usepackage[linesnumbered,algoruled,boxed,lined]{algorithm2e}
\usepackage{natbib}
\journalname{Natural Computing}

\def\makeheadbox{{%
\hbox to0pt{\vbox{\baselineskip=10dd\hrule\hbox
to\hsize{\vrule\kern3pt\vbox{\kern3pt
\hbox{\bfseries Manuscript submitted to: Natural Computing}
\hbox{This is a pre-peer review, pre-print version of this paper.}
\kern3pt}\hfil\kern3pt\vrule}\hrule}%
\hss}}}

\begin{document}

\title{A reversible system based on hybrid toggle radius-4 cellular automata and its application as a block cipher
    \thanks{This study was financed in part by the Coordenação de Aperfeiçoamento de Pessoal de Nível Superior – Brasil (CAPES) – Finance Code 001. The authors would also like to thank Conselho Nacional de Desenvolvimento Científico e Tecnológico (CNPq) and Fundação de Amparo a Pesquisa do Estado de Minas Gerais (Fapemig) for supporting this work.}
}


\author{Everton R. Lira \and 
            Heverton B. de Macêdo \and 
            Danielli A. Lima \and 
            Leonardo Alt \and 
            Gina M. B. Oliveira  
}


\institute{ Everton R. Lira \at
                Comp. Science Dept., Federal University of Uberlândia, UFU, Uberlândia, MG, Brazil \\
                (Corresponding Author) - \email{evertonlira@gmail.ccom}           
            \and
            Heverton B. de Macêdo \at
                Comp. Science Dept., Goiano Federal Institute, IF Goiano, Rio Verde, GO, Brazil 
            \and
            Danielli A. Lima \at
                Informatics Dept., Federal Institute of Triângulo Mineiro, IFTM, Patrocínio, MG, Brazil
            \and
            Leonardo Alt \at
                Berlin, Germany
            \and
            Gina M. B. Oliveira \at
                Comp. Science Dept., Federal University of Uberlândia, UFU, Uberlândia, MG, Brazil
}

\date{Received: date / Accepted: date}

\maketitle

\begin{abstract}

The dynamical system described herein uses a hybrid cellular automata (CA) mechanism to attain reversibility, and this approach is adapted to create a novel block cipher algorithm called HCA. CA are widely used for modeling complex systems and employ an inherently parallel model. Therefore, applications derived from CA have a tendency to fit very well in the current computational paradigm where scalability and multi-threading potential are quite desirable characteristics. HCA model has recently received a patent by the Brazilian agency INPI. Several evaluations and analyses performed on the model are presented here, such as theoretical discussions related to its reversibility and an analysis based on graph theory, which reduces HCA security to the well-known Hamiltonian cycle problem that belongs to the NP-complete class. Finally, the cryptographic robustness of HCA is empirically evaluated through several tests, including avalanche property compliance and the NIST randomness suite.

\end{abstract}

\section{Introduction}
\label{sec:intro}

Cryptography is the everlasting study of methods that ensure confidentiality, integrity and authentication, when storing or transmitting data, to minimize security vulnerabilities. In such approaches, data is codified in a specific way so that only those for whom it is intended can read and process it.

Despite the successful application of reputed algorithms, such as AES [\cite{daemen2002}] and RCA [\cite{menezes1996}], the evolution of hardware architectures imposes an ongoing race to develop more secure and effective encryption models. For example, with the popularization of portable electronic devices able to capture digital images, the exchange of such data between entities on private social networks or e-mails became more frequent, which led to the demand for methods that enable high throughput without loss of security.

Since the most popular symmetric encryption algorithms AES and DES are of serial nature, this poses a challenge to massive data processing [\cite{daemen:05, zeghid2007}]. This motivated a search for the improvement of these classical algorithms [\cite{prasad:13}] as an attempt to introduce parallelism on some costly or redundant steps in the process of encryption [\cite{le:10}]. However, as they are inherently sequential algorithms, this customization is limited and does not allow the desirable level of parallelism to be reached. As such, the capacity of high-performance parallel architectures can become underutilized. In this context, cellular automata (CA) appear as a useful tool in the design of inherently parallel encryption systems.

CA are totally discrete mathematical models based on the livelihood of cellular organisms, a naturally occurring process which dictates the survival of such cells based on their behavior (implemented as CA rules) while interacting with the environment conditions they are exposed to (the cell neighborhood) [\cite{rozenberg2012handbook}]. CA are widely used in the literature [\cite{sarkar2000}] and, among the most known applications, the following can be mentioned: (i) modeling of biological and physical systems [\cite{vichniac1984, ermentrout1993, maerivoet2005, alizadeh2011, ghimire2013, feliciani2016, mattei2018}]; (ii) investigation of new computational paradigms  [\cite{hillis1984, lent1993, morita2008, yilmaz2015}]; 
(iii) proposition of tools for solving various computational problems, such as task scheduling [\cite{swiecicka2006, carneiro2013, carvalho2019}], image processing [\cite{rosin2010}], computational tasks [\cite{mitchell2005, oliveira2009}], robotics [\cite{ioannidis2011, lima2017}] and, more significantly related to this article, cryptographic models [\cite{wolfram1986, gutowitz:95, sen2002cellular}].

The implementation simplicity and the ability to process data in parallel are some of the the main advantages of applying CA-based models in the most diverse areas mentioned above [\cite{vasantha2015}]. In addition, the discovery that even the simplest CA models, known as elementary, are capable of exhibiting chaotic-like dynamics [\cite{wolfram1986}], led researchers to see CA-based models as natural options for proposing fast, parallel and secure encryption methods [\cite{wolfram1986, gutowitz:95, sen2002cellular}].

A new cryptographic model called HCA (Hybrid Cellular Automata) is investigated here, which is based on chaotic one-dimensional CA rules. The HCA model recently received a patent registration in Brazil [\cite{oliveira2007sistema}], and this paper presents an unique detailed view on how the parameters of HCA were defined and on the investigations performed to validate its safety.

This model employs pre-image computation (the backward evolution of CA configuration) in the encryption process and applies one-dimensional CA rules with sensitivity to one of the extreme cells in the neighborhood, the so-called toggle rules, such as the method proposed by Gutowitz [\cite{gutowitz:95}]. Furthermore, this innovative approach also addresses two problems pointed out in relation to this previous work [\cite{gutowitz:95}]: the spread of plaintext disturbances in only one direction and the significant increment of bits in the successive pre-image computations performed during the encryption process. Moreover, several evaluations and analyses performed on the model are presented here, such as theoretical discussions related to its reversibility and an analysis based on graph theory, which reduces HCA security to the well-known Hamiltonian cycle problem. The cryptographic robustness of the model is empirically evaluated through some security analysis, such as, the avalanche property compliance and the NIST randomness suite.

Even though later models, also based on CA toggle rules, sought to reduce these problems [\cite{wuensche2008, oliveira:04, oliveira:08, oliveira:10, silva2016}], the solution proposed here is the only one that guarantees an appropriate propagation of the disturbance over the entire lattice, as well as keeping the size of the ciphertext the same as the plaintext. Aiming at ensuring a good perturbation propagation, the model investigated here uses a lattice with a periodic boundary condition, which allows a simple 1-bit disturbance to be propagated throughout the entire lattice, regardless of the position where this perturbation occurs and the direction of the rule sensitivity. 

Moreover, to make sure there is always a pre-image for any lattice configuration and also to avoid the bit increase seen in previous models, the HCA model is heterogeneous and applies two distinct rules, both of which are toggle rules sensitive to the same direction. The so-called main rule possesses chaotic dynamics and is used on most bits of the lattice, whereas the so-called border rule possesses fixed-point dynamics with spatial displacement and is used on a small number of consecutive bits, which are called the lattice border. At the same time that the main rule guarantees the injection of appropriate entropy in the lattice as the pre-images are calculated, the border rule guarantees the existence of a single pre-image for each configuration. 

In addition, the position of the cells that are characterized as the border of the lattice vary for each consecutive pre-image computation in order to assure that every cell was evolved by the main chaotic rule at various steps. This variation is also made in order to promote a higher level of parallelism in the ciphering stage. Using the scheme proposed in HCA, each subsequent pre-image computation can be started soon after the initial bits of the previous pre-image calculus are known.

Section \ref{sec:related} presents a review of the main works in the literature related to the investigated method. Section \ref{sec:proposedsolution} formally presents the HCA model and details all of the processes involved in its proposition. Section \ref{sec:reversibility} presents a theoretical aspect related to HCA: the proof that the hybrid CA model used in HCA is reversible, unlike the model used by Gutowitz [\cite{gutowitz:95}] that is irreversible. Section \ref{sec:formal} presents a formal analysis of the model based on graph theory, which associates the problem of breaking the HCA key with the problem of finding a Hamiltonian cycle in a graph, which belongs to the NP-complete class. Section \ref{sec:security} describes several analyses established in the literature to verify the security of a cryptographic method, presenting the suitability of each one in the evaluation of HCA. The experimental results obtained in three of the analyses described in section 6 are presented in section \ref{sec:experiments} for the validation of HCA security against cryptanalysis attacks: plaintext avalanche effect, key avalanche effect and NIST suite tests. Finally, Section \ref{sec:conclusion} presents the main conclusions of our investigation about the HCA cryptographic method and proposes some future directions in this research.

\section{Related Work}
\label{sec:related}

The first suggestion about the employment of cellular automata models in cryptography was made by Wolfram [\cite{wolfram1985}], after his studies on the statistical properties of CA chaotic rules with radius 1, which can be used as pseudo-random number generators [\cite{wolfram1986}]. Since then, various studies on this topic have been made [\cite{tomassini2000, sen2002cellular, vasantha2015, oliveira2007sistema, benkiniouar2004, nandi1994, gutowitz:95, oliveira:04, oliveira:08, oliveira:10, silva2016, wolfram1985, wuensche2008, oliveira_icsc:10, wuensche:92, seredynski2004, yang2016novel}], where the cryptographic models can be classified into three kinds of approaches. 

The first approach, proposed by Wolfram, takes advantage of the good pseudo-random properties of known transition rules with chaotic behavior to generate random binary sequences. Therefore, the rules are not used as the cryptographic key, which, in fact, corresponds to the initial lattice. This lattice is evolved by a pre-specified chaotic rule (elementary rule 30) and the sequence of bits generated in a specific cell is used as a pseudorandom sequence. Moreover, the effective ciphering process is made by a reversible function that mixes the plaintext with the random sequence, such as the XOR logical function [\cite{wolfram1985, wolfram1986, tomassini2000, benkiniouar2004, nandi1994}]. On the contrary, the HCA model discussed here employs transition rules as secret keys and the initial lattice corresponds to the plaintext. More recently, this approach was diversified, for example, by using different one-dimensional transition rules with radius 1 and 2 and also two-dimensional rules, and using  evolutionary search for finding suitable chaotic rules [\cite{seredynski2004, tomassini2001, sirakoulis2016, toffoli1987, kari1992, machicao2012chaotic, john2020design}]. Another line of investigation is the parallelization of cellular automata as pseudo-random number generators that can be applied in cryptographic schemes [\cite{sirakoulis2016}]. 

The second approach is based on additive, non-homogeneous and reversible CA rules. The cryptographic keys are typically a combination of known additive rules [\cite{toffoli1987}] that exhibit algebraic properties. When such rules are used together in a heterogeneous scheme, they exhibit a periodic dynamics with maximum and/or known cycle [\cite{nandi1994, kari1992}]. However, the parallelism and safety of these models are limited, due to the additive property of the rules, which prevents the chaoticity of the rules. The system proposed in [\cite{nandi1994}] was broken in [\cite{blackburn1997}] by analyzing the additive properties of the rules. More recently proposed systems based on this line of research have been mixing additive rules and nonlinear rules to circumvent the security problems of their predecessors [\cite{das2010generating}].

The last approach uses the backward evolution of the CA lattice to cipher the plaintext. The cryptographic key is the CA transition rule and it must have some properties to ensure the pre-image existence [\cite{oliveira:08, wuensche2008, oliveira_icsc:10, wuensche:92}]. Gutowitz was the first to propose a cryptographic model using such approach; it is based on the backward evolution of irreversible homogeneous CA [\cite{gutowitz:95}]. The cryptographic model discussed here also uses the backward evolution. However, in the novel HCA method, the rules are reversible and they are applied in a scheme where two different rules are used to ensure the pre-image existence, defining a heterogeneous CA model. Therefore, we further detail the state of art related to CA-based models that belong to the third approach.

Gutowitz’s model employs CA toggle rules, which are used as cryptographic keys (or a part of these). Such rules are sensitive to the leftmost and/or to the rightmost cell in the neighborhood. That means any modification to the state of this cell necessarily causes a modification on the central cell. A pre-image of an arbitrary lattice of size N is calculated adding R extra bits to each side and a pre-image will be calculated with N + 2R  cells. If a right-toggle rule transition is used as key, the pre-image cells can be obtained in a deterministic way, step-by-step, from the leftmost side to the right [\cite{gutowitz1994method}]. The plaintext corresponds to the initial lattice and P pre-images are calculated to obtain the ciphertext. As 2R bits are added to each pre-image calculated, the size of the final lattice is given by N + 2RP. Such non-negligible increment is pointed as the major drawback of this model. Moreover, another flaw was identified in it, a high degree of similarity between ciphertexts was observed when the plaintext is submitted to a little perturbation. To deal with this problem, the model employs two phases where a left-toggle and a right-toggle rule are applied in each stage. Both rules are generated starting from the same cryptographic key, however, it needs more time steps to cipher the plaintext. Later on, this model was altered by using bidirectional toggle CA rules (to the right and to the left simultaneously) in [\cite{oliveira:04}], showing that the similarity flaw was solved with such a modification and that it is protected against differential cryptanalysis. However, the ciphertext increment in relation to the plaintext length remains in this model.

An algorithm known as reverse algorithm was proposed in [\cite{wuensche:92}] for a pre-image computation starting from any lattice and applying an arbitrary transition rule (not only toggle rules). However, using a periodic boundary CA, the pre-image computation is concluded verifying whether the initial bits can be equal to the final 2R rightmost ones. If so, the extra bits are discarded returning the pre-image to the same size of the original lattice. If no, this pre-image does not exist. This algorithm finds all the possible pre-images for any arbitrary periodic boundary lattice, if at least one exists. This reverse algorithm was evaluated as an encryption method in [\cite{oliveira:08}] and [\cite{wuensche2008}]. However, since there is no guarantee of pre-image existence for all possible rule transitions, the major challenge in these previous models was to evaluate the characteristics of the rules to assure the existence of at least one pre-image for any possible lattice. An attempt to solve this problem was to use the Z parameter [\cite{silva2016}] in the rule specification. The method proposed in [\cite{wuensche2008}] is very similar to the initial method proposed in [\cite{oliveira:08}], despite being developed independently. The major conclusion in [\cite{wuensche2008}] is that the simple adoption of the reverse algorithm is not viable because the possible rules with 100\% guarantee of pre-image existence are not appropriate for ciphering, even when using the Z parameter to choose suitable secret keys. No treatment to this problem was addressed in [\cite{oliveira:08}], that is, how to proceed if a failure occurs when computing pre-images. It is an important point to discern the works in [\cite{wuensche2008}] and [\cite{oliveira:08}]. An alternative approach to use the reverse algorithm by adopting a wrap procedure was later investigated in [\cite{oliveira_icsc:10}]. This contour procedure ensures any plaintext can be encrypted. However, it generates a variable size ciphertext, which can be larger than the plaintext. Later on, it was shown that an appropriate specification of the  secret key gives a low probability to this failure occurrence [\cite{oliveira:10}], expecting to rarely apply this contour procedure keeping the ciphertext size close to the plaintext. This specification was deeper investigated in [\cite{oliveira2010_ppsn}] and [\cite{oliveira2011deeper}]. Additionally, a cryptographic model that employs a lattice with a fixed extra boundary was investigated in [\cite{silva2016}], which applies the reverse algorithm proposed by Wolfram. Even though the lattice increase is smaller than in Gutowitz’s model, the final lattice is still larger than the plaintext, which increases the cost of sending encrypted information, in addition to the aperiodic condition of the lattice hindering the good propagation of disturbances.

As far as we know, the first CA model that uses backward evolution with chaotic toggle rules and that has 100\% of preimage calculus keeping the ciphertext with the same size of the plaintext (using a periodic boundary condition) is the one discussed in this paper. HCA was first proposed in [\cite{macedo2007}] and a patent registration was submitted to the Brazilian agency of patents (INPI) in 2007, which has been recently accepted in 2019 [\cite{oliveira2007sistema}]. Meanwhile, other academic works have investigated different aspects of HCA and propose some adaptations of this CA-based model [\cite{magalhaes2010metodo, lima2012modelo, alt2013propriedades}]. Some of the analyses over HCA investigated in these works are presented here. More recently, a new model inspired by HCA was proposed, replacing the cellular automata structure by complex networks connections [\cite{macedo2014}]. In spite of some advantages related to the fast propagation of information promoted by non-local connections, the intrinsic parallelism of CA models is not presented in the model based on complex networks.

\section{HCA Method Description}
\label{sec:proposedsolution}

The HCA method consists of a symmetric block-based cryptographic system that uses the dynamic behavior of CAs to perform the cipher and decipher process. Both forward and backward (pre-image) evolution of CAs are essential parts of this algorithm.

\subsection{HCA - Block Size Definition}
\label{sub:proposedSolution_blockSizeDefinition}

In HCA, 128-bit blocks are used for the cipher and decipher process. The method could be easily adapted to be used with other block sizes, but this value was set to conform with the current standard for symmetric cryptography methods.

Like all block-cipher methods, the HCA method is compatible with every mode of operation described in the literature, such as ECB, CBC, OFB, CFB, CTR, among others [\cite{NIST2014}]. Despite this, the use of ECB and CBC are discouraged due to the publicly known inherent vulnerabilities caused by ECB [\cite{rogaway2011evaluation}] and to the existence of padding oracle attacks applicable when using CBC [\cite{vaudenay2004security}].

\subsection{HCA - Cryptographic Key Definition}
\label{sub:proposedSolution_cryptographicKeyDefinition}

The HCA cryptographic key ($ K $) is formed by a 257-bit sequence. 
Moreover, the initial 256 bits of the cryptographic key are used to produce radius-4 CA transition rules ($ r=4 $), which have $ 512 $ bits. An explanation on how CA rules are derived from keys is provided in section \ref{subsec:generating_rules}. 

The total space of possible cryptographic keys is formed by $ 2^{256} $ left-toggle rules and other $ 2^{256} $ right-toggle rules. However, some of them are discarded for this approach because they do not produce the desired dynamic behavior. The normalized spatial entropy ($ h $) calculated on the initial 256 bits of a potential key must be greater than 0.75 ($ h > 0.75 $) for it to be considered a suitable key. It is calculated by the Expression (\ref{eq_1}), where $ p_i $ is the probability of an 8-bit substring occurring in the 256-bit sequence, which is evaluated for every possible 8-bit binary combination through the summation.

\begin{equation}
\label{eq_1}
h=\frac{-\sum_{i=1}^{256} {(p_i \times \log_{2}(p_i))}}{8}
\end{equation}

Setting the cryptographic key entropy above $ 0.75 $ causes the HCA method to generate cellular automaton toggle rules with chaotic dynamics as shown in [\cite{oliveira2010_ppsn}]. This kind of rule does not have an easily identifiable pattern during CA evolution and therefore makes the resulting ciphertext enough harder to decipher when the cryptographic key is not known.

On the contrary, keys with $ h\leq0.75 $ are discarded. This reduction of the valid key space is estimated to be quite low when compared to the total potential key space ($ 2 \times 2 ^ {256} $). Table \ref{subsec:ValidKeys_relation} presents the relation between the amount of bits used in the key ($ \vert K \vert $) and the percentage of discarded keys.

\begin{table}[!ht]
	\centering
	\caption{Key bits amount $\times$ Discarded keys percentage}
	\label{subsec:ValidKeys_relation}
	\begin{tabular}{|c|c|p{2.5cm}|p{2.5cm}|}
		\hline
		\textbf{Radius} & \textbf{$\vert K \vert$}   & \textbf{Keyspace}  & \textbf{Discarded (\%)}  \\ \hline
		\textbf{1}	&\textbf{4-bits} 	& $2\times{2}^{4} = 32$			& 25 \%  \\ \hline
		\textbf{2}	&\textbf{16-bits}	& $2\times{2}^{16} = 131072$	& 8.64 \%  \\ \hline
		\textbf{3}	&\textbf{64-bits}	& $2\times{2}^{64}$				& $\approx 0.113$ \%  \\ \hline
		\textbf{4}	&\textbf{256-bits}	& $2\times{2}^{256}$			& $\approx < 0.1 \times {10}^{-8}$ \%  \\ \hline
	\end{tabular}
\end{table}

In Table \ref{subsec:ValidKeys_relation} the percentages listed for $ r = 1 $ and $ r = 2 $ are absolute, since all the possible keys were tested. An analysis of the entire keyspace for $ r \geq 3 $ is impractical, but extrapolations based on random sampling are presented for $r = 3$ and $r = 4$. For both estimates, ${2}^{32}$ keys were randomly generated and evaluated in regard to the acceptance criteria ($ h > 0.75 $). An apparent correlation can be observed between increasing the CA radius and a reduction in the percentage of discarded keys. The estimated $r = 4$ discarded keys percentage suggests that only a minimal set of very homogeneous keys would be rejected in the vast $2\times{2}^{256}$ keyspace.

\subsection{HCA - Defining Operations}
\label{sub:proposedSolution_definingOperations}

Considering each 128-bit block, the binary sequence is the initial lattice configuration, $t = 0$, for the CA. The cipher procedure equals to applying the reverse evolution operation ($ \lambda $), also known as pre-image calculus, for 128 steps until the configuration $t = -128$ is reached. Two CA rules derived from the cryptographic key ($ K $) are applied at each step. The decipher process is performed through the forward CA evolution operation ($ \phi $), using the same set of rules employed in the cipher procedure.

\subsubsection{Generating CA Rules from a Key}
\label{subsec:generating_rules}

By definition, cellular automata rules can be expressed through mappings from the CA neighborhood bit values to the single bit result value. The HCA method employs a specific subset of CA rules (also called toggle rules) which ensures that every configuration will have a single pre-image, and that this pre-image can be calculated in a deterministic manner through the lattice reverse evolution. Mappings of toggle rules display a characteristic of the result bit being sensible to value changes in either extremity (or both extremities) of the CA neighborhood. So, considering a left-toggle CA rule, in all mappings specified by this rule, value changes in the leftmost bit of the neighborhood will change the output central bit for the new resulting configuration. Similarly, distinct values on the rightmost bit in a CA neighborhood will surely provide distinct resulting bit values when a right-toggle rule is applied. Figure \ref{fig-sensibilidadeRegra} provides examples of this concept.

\begin{figure}[!ht]
  \begin{center}
    \includegraphics[width=0.7\columnwidth]{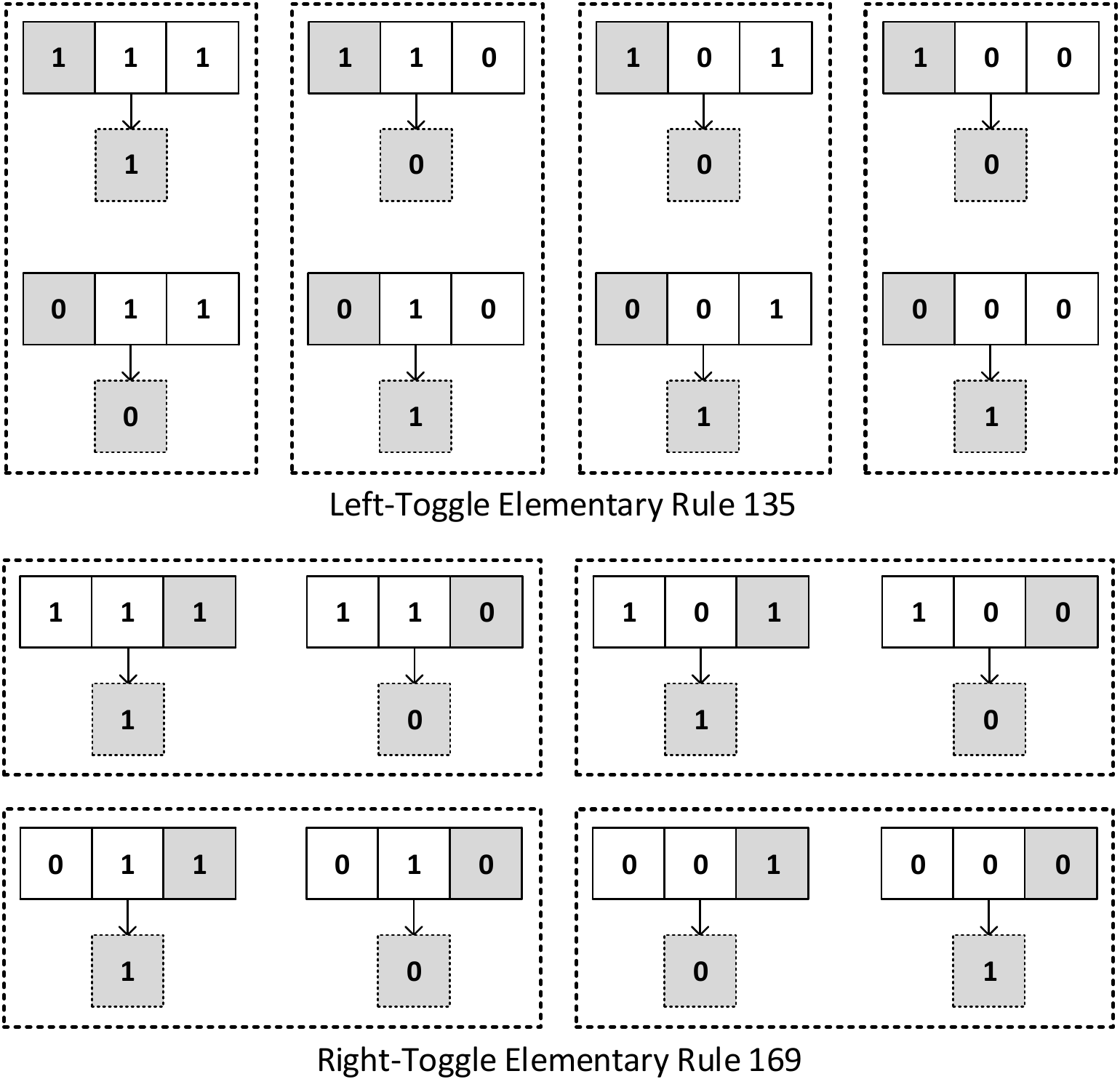}
  \end{center}
  \caption{Example of CA toggle rules.}
  \label{fig-sensibilidadeRegra}
\end{figure}

The CA neighborhood mappings specified by elementary rules 135 and 169 are displayed in Figure \ref{fig-sensibilidadeRegra}. Rule 135, presented at the top, is a left-toggle rule and, as such, mappings in which the neighborhood differs only on its leftmost bit value will always result in distinct output bit values. Likewise, for the right-toggle rule 169, value changes in the rightmost bit of the neighborhood will alter the resulting bit.

Taking into account the speed of commercially available hardware, HCA employs radius-4 rules ($ r = 4 $) to ensure a large keyspace ($2\times{2}^{256}$) deemed appropriate against a brute force attack. Radius-4 rules are made up of 512 bits, but toggle rules of this radius can be derived from any 256-bit sequence. In the HCA method, the last digit of the 257-bit cryptographic key is responsible for defining the toggle direction of the employed rules and the first 256-bits are used to generate radius-4 toggle rules with the desired dynamic behavior. Two CA rules are used at each evolution step, which will now be defined as the main rule ($ \phi_{m} $) and border rule ($ \phi_{b} $). 

Consider a cryptographic key $K$ of 257-bits, so that $ K = K[0], K[1], ..., K[256] $. The $ K[256] $ bit determines the toggle-direction of the generated rules. If $ K[256] = 0 $, left-toggle rules will be generated, otherwise right-toggle rules are produced. The 512-bit CA main rule, $ \phi_{m} $, is derived from $ K $ using Expression (\ref{generateRP}), where the $ + $ sign stands for concatenation and the upper slash indicates a binary complement operation.

\begin{gather}
createK_{m}(K) = \nonumber \\
\label{generateRP}
\left \{ 
    \begin{matrix} 
        K + \overline{K}, & K[256] = 0 \\
        K[0], \overline{K[0]}, \cdots, K[255], \overline{K[255]}, & K[256] = 1 
    \end{matrix} 
\right.
\end{gather}

While $ \phi_{m} $ is derived from the initial 256 bits of $ K $ as shown in Expression (\ref{generateRP}), the 512-bit border rule, $ \phi_{b} $ is selected in a subset of only four rules. The subset is formed by two left-toggle rules $ \{11\cdots1+00\cdots0\}, \{00\cdots0+11\cdots1\} $ and two right-toggle rules $ \{1010\cdots10\}, \{0101\cdots01\} $. Two criteria guide the selection: (1) the first bit of border rule $ \phi_{b}[0] $ must be the complement value of the first bit of main rule $ \phi_{m}[0] $, and (2) the border rule will share the same toggle direction as the main rule. The expression (\ref{generateRB}) define the selection criteria.

\begin{gather}
createK_{b}(K) = \nonumber \\ 
\label{generateRB}
\left \{
    \begin{matrix} 
        11\cdots1 + 00\cdots0, & \phi_{m}[0]=0, K[256]=0 \\
        1010\cdots10, & \phi_{m}[0]=0, K[256]=1 \\
        00\cdots0 + 11\cdots1, & \phi_{m}[0]=1, K[256]=0 \\
        0101\cdots01, & \phi_{m}[0]=1, K[256]=1 \\
    \end{matrix} 
\right.
\end{gather}

These four rules are a unique subset of toggle rules for which there is a direct relation between the bit in the toggle direction extremity of the neighborhood and the output bit; this interrelation is so absolute that the output bit value can be determined regardless of the values in other bits of the input neighborhood, and it is crucial to the pre-image calculus procedure.

\subsubsection{Backward Evolution Operation}

Given an lattice $ s $ at step $ t $ (represented as $ s^{t} $), consider the backward evolution operation (pre-image calculus) as $ \lambda (s^{t}, \phi_{m} , \phi_{b}) = s^{t-1} $. Applying $ \lambda $ to $ s^{t} $ means finding all bits of $ s $ at time $ t-1 $ using the rules $ \phi_{m} $ and $ \phi_{b} $, where $ s^{t-1} = s^{t-1}[0], s^{t-1}[1], \cdots, s^{t-1}[127] $.

Considering a radius-4 rule, the pre-image calculus begins by determining the value of 8 consecutive bits ($b1, b2, b3, b4, b5, b6, b7, b8$) of the pre-image using $ \phi_{b} $. As previously stated, for any bit $s^{t}[i]$ calculated through $ \phi_{b} $ there is a relation of value equality or complement between it and the bit in the relevant extremity of the neighborhood in $t-1$. So by knowing which of these rules is $\phi_{b}$, and assuming it was applied to the $s^{t-1}[i]$ neighborhood, if $\phi_{b}$ is a known left-toggle rule and the value of $s^{t}[i]$ is also available, then $s^{t-1}[i-4]$ can be determined, as shown in Expression (\ref{bitPreImageBorderLeft}); or if $\phi_{b}$ is a known right-toggle rule, 
then $s^{t-1}[i+4]$ can be determined, as displayed in Expression (\ref{bitPreImageBorderRight}).

\begin{equation}
\label{bitPreImageBorderLeft}
s^{t-1}[i] = \lambda_{\phi_{b}}(s^{t}[(i+4) \mod 128])
\end{equation}

\begin{equation}
\label{bitPreImageBorderRight}
s^{t-1}[i] = \lambda_{\phi_{b}}(s^{t}[((i-4) + 128) \mod 128])
\end{equation}

This procedure is used to determine 8 successive bits of the pre-image in $ t-1 $, which are regarded as the border region; each bit calculus depends only on the value of another single cell at step $ t $. This is only possible due to the simplicity of border rule $\phi_{b}$ that imposes a non-chaotic dynamic behavior to these eight cells. Such non-chaotic behavior will not affect the quality of the algorithm, since the border rule will not have a significant influence on the CA dynamic as a whole. The border rule is only used to ensure the existence of a single pre-image for any possible configuration, as proved in Section \ref{sec:reversibility}.

All the other 120-bits of the pre-image are obtained from the main rule mapping $ \phi_{m} $, responsible for provides the desired chaotic behavior to the algorithm. The values of these 120 remaining bits are determined, one by one, in the order displayed in Figure \ref{fig-ordemCalculoPreImagem}.

\begin{figure}[!ht]
  \begin{center}
    \includegraphics[width=0.7\columnwidth]{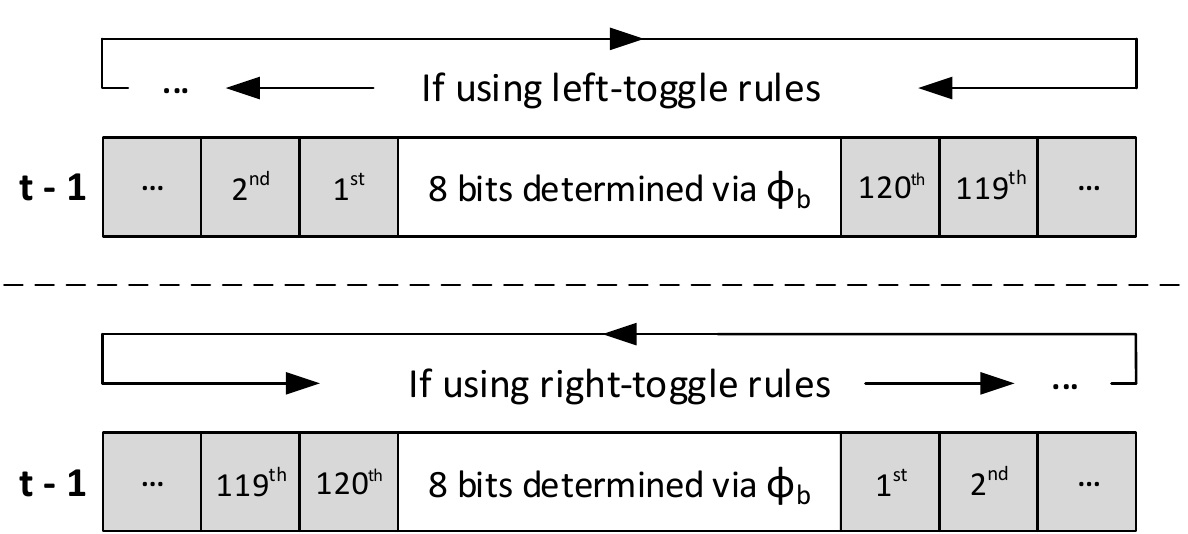}
  \end{center}
  \caption{Order of main bits determination according to toggle direction.}
  \label{fig-ordemCalculoPreImagem}
\end{figure}

If the bit determination order presented in Figure \ref{fig-ordemCalculoPreImagem} were applied to a situation where the last bit of the key $K$ defines the HCA execution toggle direction as ``left'', then the calculation of the first bit using the left-toggle main rule, $ \phi_{m} $, would be as represented in Figure \ref{fig-firstMainBitDeterminationLeft}.

\begin{figure}[!ht]
  \begin{center}
    \includegraphics[width=0.7\columnwidth]{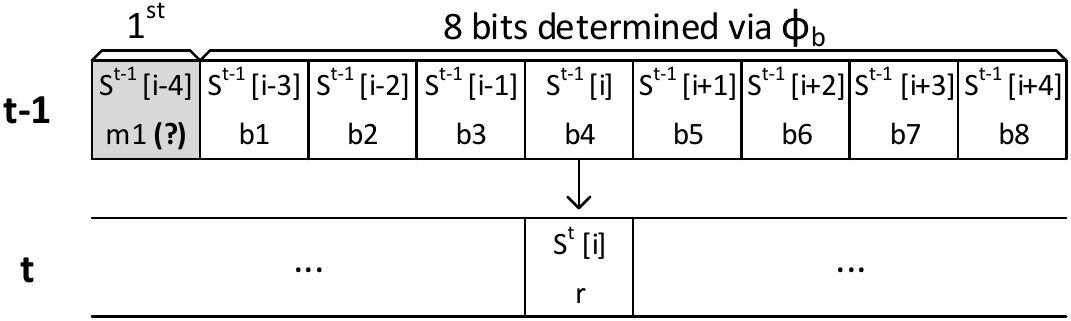}
  \end{center}
  \caption{First main bit determination for left-toggle rule.}
  \label{fig-firstMainBitDeterminationLeft}
\end{figure}

An initial supposition for the context presented in Figure \ref{fig-firstMainBitDeterminationLeft} is that the bit value in position $s^{t}[i]$ has been determined by main rule $ \phi_{m} $. Therefore, there is a valid mapping, specified by $ \phi_{m} $, from the ($m1, b1, \cdots, b8$) values in the $s^{t-1}[i]$ neighborhood ($s^{t-1}[i-4], s^{t-1}[i-3], \cdots, s^{t-1}[i+4]$) to the output value in $s^{t}[i]$, which is $r$.

Since the main rule, $ \phi_{m} $, is a proper radius-4 CA rule, it provides mappings from all possible 9-bit neighborhood combinations to their corresponding output bits, including for $(m1, b1, b2, b3, b4 , b5, b6, b7, b8) = (r)$. Due to the characteristics of the rule $ \phi_{m} $, there is only one value that $m1$ could assume in Figure \ref{fig-firstMainBitDeterminationLeft}. This deterministic procedure is listed in Expression (\ref{bitPreImageMainLeft}) for left-toggle rules and in Expression (\ref{bitPreImageMainRight}) for right-toggle rules.

\begin{equation}
\label{bitPreImageMainLeft}
\begin{matrix}
S^{t-1}[i] = \lambda_{\phi_{m}}(s^{t}[(i+4) \mod 128]] , \\ 
            s^{t-1}[(i+1) \mod 128], \\
            s^{t-1}[(i+2) \mod 128], \\
            \cdots, \\
            s^{t-1}[(i+8) \mod 128])
\end{matrix}
\end{equation}

\begin{equation}
\label{bitPreImageMainRight}
\begin{matrix}
S^{t-1}[i] = \lambda_{\phi_{m}}(s^{t}[((i-4) + 128) \mod 128]] , \\ 
            s^{t-1}[((i-1) + 128) \mod 128], \\
            s^{t-1}[((i-2) + 128) \mod 128], \\
            \cdots, \\
            s^{t-1}[((i-8) + 128) \mod 128])
\end{matrix}
\end{equation}

After determining the $m1$ value for position $s^{t-1}[i-4]$, it can be used to determine $m2$ and, progressively, to determine every one of the 120 main bits of the pre-image. The first steps, and their respective considered neighborhoods, for this operation are presented on Figure \ref{fig-firstMainBitsDeterminationLeft}.

\begin{figure}[!ht]
  \begin{center}
    \includegraphics[width=0.85\columnwidth]{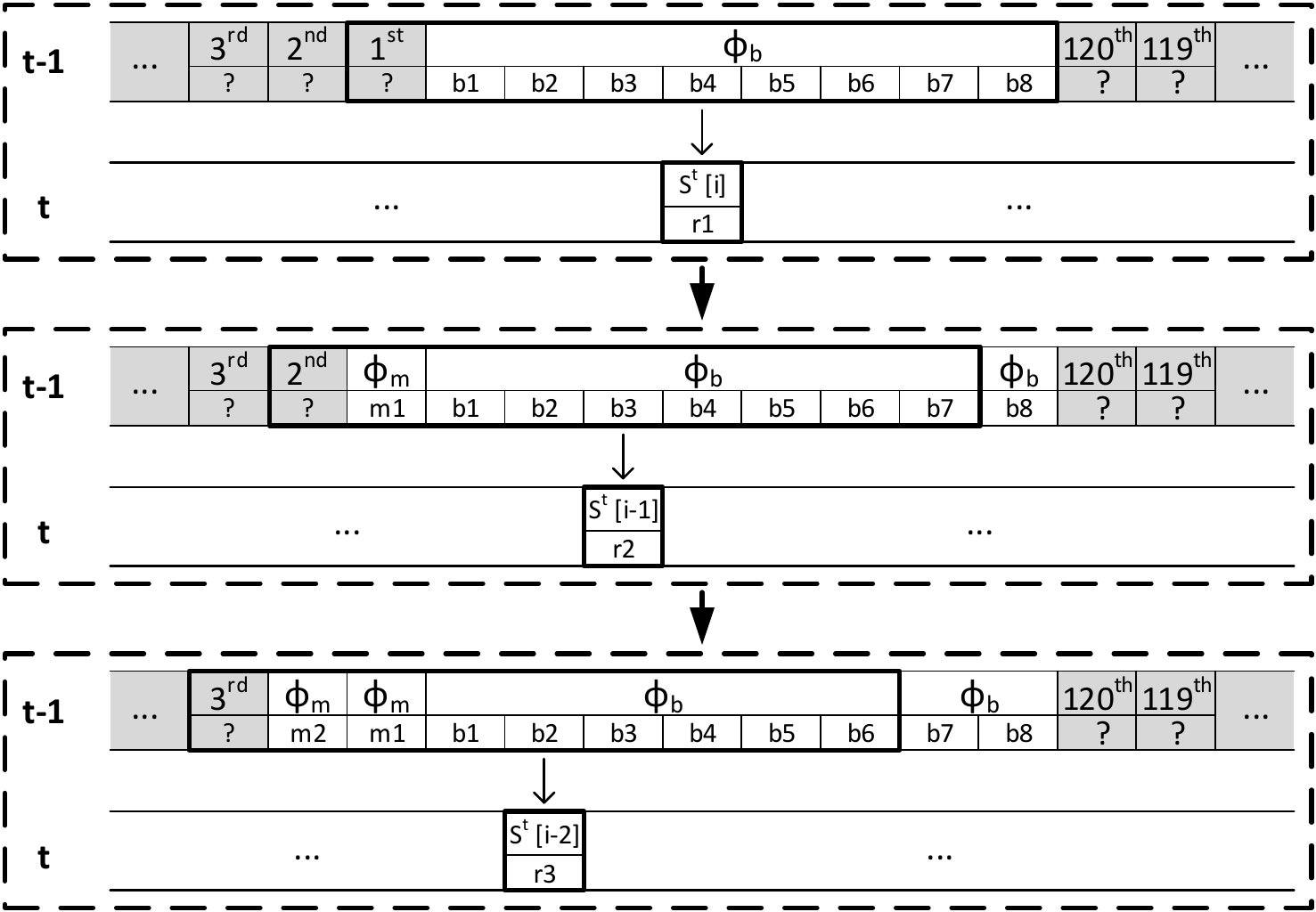}
  \end{center}
  \caption{Next main bits computation for left-toggle rule.}
  \label{fig-firstMainBitsDeterminationLeft}
\end{figure}

The equivalent procedure when using right-toggle rules is easily derivable from the left-toggle rules example, since the main difference is the order in which bits are evaluated, as displayed in Figure \ref{fig-ordemCalculoPreImagem}. The first backward evolution operation ($\lambda$) is concluded after 120 main bits of the pre-image are evaluated, since this means all the 128 bits in configuration $s^{t-1}$ have been determined.

\subsubsection{Forward Evolution Operation}

Given a lattice $ s $ at step $ t $, consider the forward evolution as $ \Phi (s^t, \phi_{m}, \phi_{b}) = s^{t+1} $. Applying $ \Phi $ to $ s^t $ means finding all bits of $ s $ at step $ t+1 $, using the rules $ \phi_{m} $ and $ \phi_{b} $, where $ s^{t+1} = s^{t+1}[0], s^{t+1}[1], \cdots, s^{t+1}[127] $. For each position $i$, the bit value of the cell $ s^{t+1}[i] $ is updated by a rule mapping from the 9-bit neighborhood in $ s^{t}[i] $ considering the main rule ($ \phi_{m} $) or the border rule ($ \phi_{b} $). 

In this forward evolution procedure all cells in $s^{t+1}$ can be determined simultaneously. From the 128 cells, 8 are updated using the border rule, $\phi_{b}$, and 120 cells are updated using the main rule, $\phi_{m}$. A relevant listing of which cells are updated by each rule is provided in Section \ref{sub:proposedSolution_latticeRegions}. The operation $ \Phi (s^t, \phi_{m}, \phi_{b}) $ can be represented by $ \Phi (s^t) $, as displayed in Expression (\ref{bitsEvolutionGeneralDetails}).

\begin{equation}
\label{bitsEvolutionGeneralDetails}
\Phi(s^t, \phi_{m}, \phi_{b}) = \Phi(s^t) = s^{t+1}
\end{equation}

\subsection{HCA - Parallelism and Lattice Regions}
\label{sub:proposedSolution_latticeRegions}

A relevant characteristic of cellular automata is the inherent parallelism of these systems. In a conventional forward evolution procedure, all the cells in a certain lattice $s$ at time step $t$ can be evolved simultaneously to generate the $s^{t+1}$ configuration. Since the decryption process of HCA is based on the forward evolution operation ($\Phi$), with proper hardware it is possible to evolve all 128 cells from the $ s^t $ lattice to the $ s^{t+1} $ lattice in parallel, with a considerable performance gain.

On the other hand, since the encryption process of HCA is based on the backward evolution operation ($\lambda$), a distinct way of achieving parallelism was devised. The HCA encryption is based on applying 128 successive pre-image calculus operations to each block, and, conventionally, the calculus of a $s^{t-1}$ pre-image would only be started after all the bits of $s^{t}$ are known. In this scenario, parallelism is attained by making it possible that bits from distinct pre-images are determined simultaneously.

Expression (\ref{bitsEvolucaoGeral}) indicates which cells are evolved with the main rule and which use the border rule.

\begin{equation}
\label{bitsEvolucaoGeral}
 \Phi(s^t[i], \phi_{m}, \phi_{b}) = \left \{ \begin{matrix} 
        \phi_{b}(s^t[i]), i = \{0, 1, \cdots, 7\} \\ 
        \phi_{m}(s^t[i]), i = \{8, 9, \cdots, 127\}
        \end{matrix} \right.
\end{equation}

According to Expression (\ref{bitsEvolucaoGeral}), 8 cells, $ s[0], s[1], \cdots, s[7] $, are evolved with the border rule ($ \phi_{b} $), and the remaining 120 cells, starting from $ s[8] $ to $ s[127] $, through the main rule ($ \phi_{m} $).

To attenuate any impact of the non-chaotic dynamic behavior of border rules in the quality of the algorithm, the HCA method uses an 8-bit circular shift from the $ s $ block in the opposite direction to the toggle direction set by the last bit of the cryptographic key ($ K [256] $). The Expression (\ref{eq:rotacao8}) defines the operation to be performed between each pre-image calculus.

\begin{equation}
\label{eq:rotacao8}
s \leftarrow Shift_8(s, opposite(toggle\_direction(K[256])))
\end{equation}

Considering this 8-bits circular shift mechanism and that one processing cycle is enough for each cell evaluation when all values needed for such evaluation are available, Figure \ref{fig-HCAparallelism} presents at which processing cycle the first bits of the initial pre-images would be evaluated for a left-toggle execution. 
\begin{figure}[!ht]
    \centering
    \includegraphics[width=0.75\columnwidth]{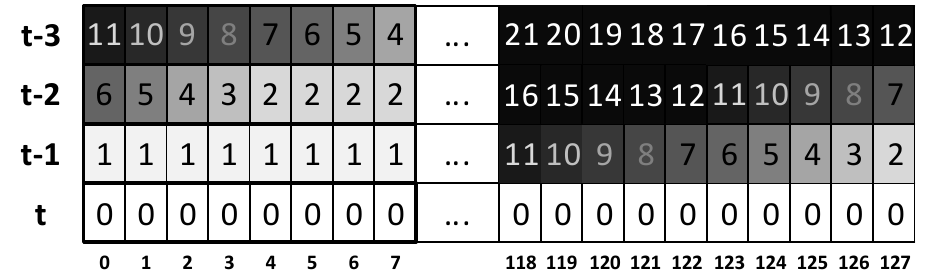}
    \caption{Processing cycles in which pre-image cells are evaluated.}
    \label{fig-HCAparallelism}
\end{figure}

In Figure \ref{fig-HCAparallelism}, lighter color tones indicate cells that are evaluated first, and the number inside each cell is the processing cycle at which the cell value is determined. Through this approach, about 629 cycles would be needed in specialized hardware, to determine the $t - 128$ pre-image, instead of the $128^2 = 16.384$ cycles needed in a purely sequential approach.

\subsection{HCA Method Overview}

The cryptographic key $ K $, used in the process of generating the $ \phi_{m} $ and $ \phi_{b} $ rules, must be applied in a way that generates the same rules for equivalent cipher and decipher steps. Thus, the key $ K $ used at step $ t=1 $ in the cipher process should be the same as the one used at step $ t=128 $ in the decipher process. A scheme illustrating this relation is presented in Figure \ref{fig:detailsHCA}, where the left side displays the encryption process being performed from top to bottom, and the right side of the figure shows the decryption process being performed in the opposite direction.

\begin{figure}
    \centering
    \includegraphics[width=0.8\columnwidth]{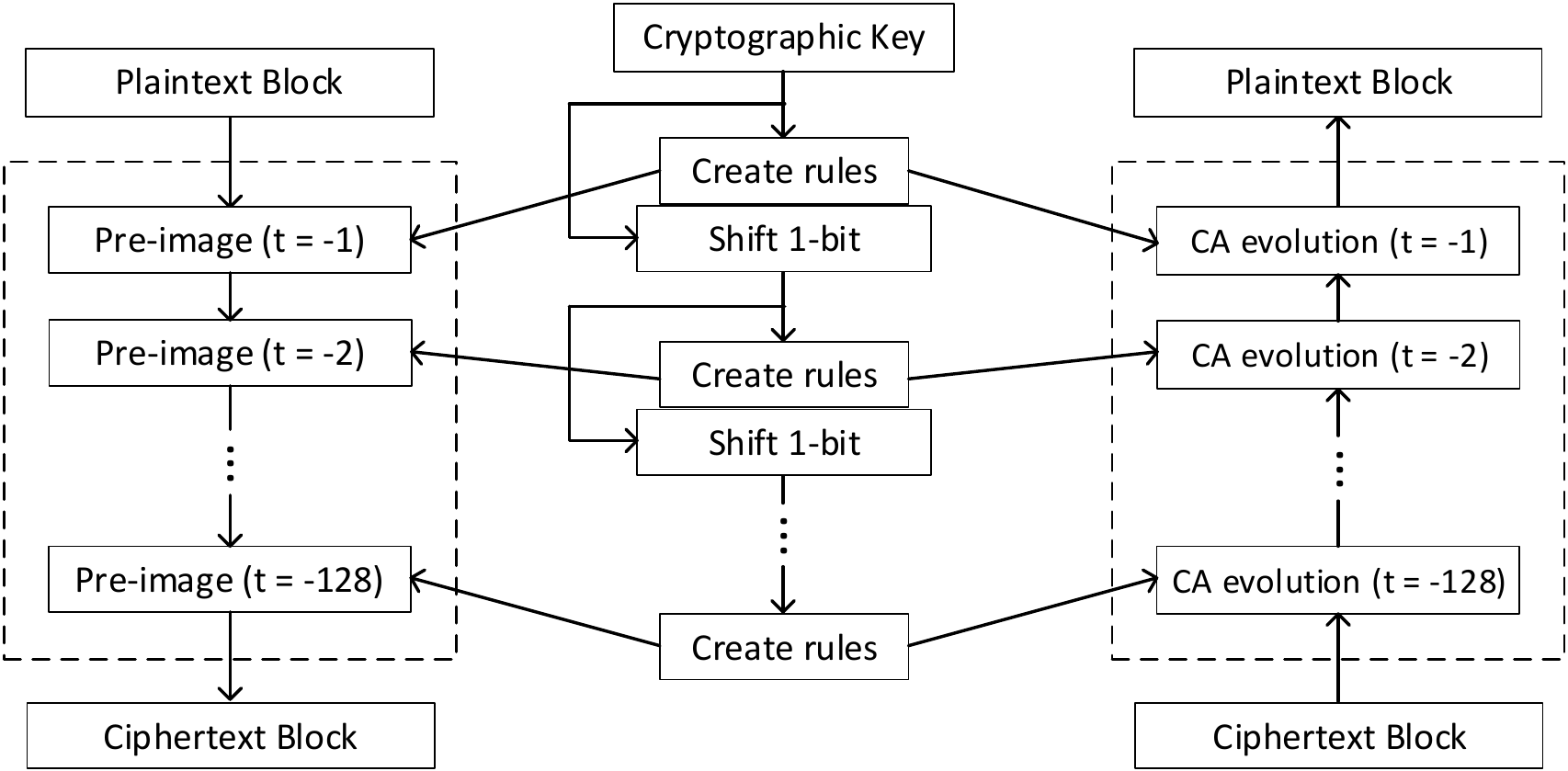}
    \caption{Scheme illustrating the HCA cipher and decipher process.}
    \label{fig:detailsHCA}
\end{figure}

During the ciphering process, at each CA step, the cryptographic key $ K $ is shifted to the left by 1 bit, generating two new rules $ \phi_{m} $ and $ \phi_{b} $. 

It is important to note that, in the deciphering process, the encrypted block ($s^{-128}$) will be used as the initial configuration for the procedure, but the cryptographic key used in this step will be $ Shift_{127}(K, left) $. The original cryptographic key $K$ must be shifted to the left by 127 positions before rules $ \phi_{m} $ and $ \phi_{b} $ are derived from it. At each decryption step it will be necessary to shift this cryptographic key obtained in the previous step to the right by 1 position, so that it becomes equivalent to the key used in the corresponding encryption step.

It is expected that all main rules derived from $K$ in the cipher process have chaotic dynamic behavior, and since many distinct rules are employed in the method, it is harder for a cryptographic attack to exploit the dynamic behavior of a specific rule.

The Algorithm \ref{alg:cipher} presents in pseudo-code the operations performed during the encryption process (backward evolution) and the Algorithm \ref{alg:decipher} shows the decryption process operations (forward evolution).

\begin{algorithm}[!ht]
\label{alg:cipher}
\SetKwInput{Input}{input}
\SetKwInOut{Output}{output}
\Input{A key $K$ and plaintext block $s$}
\Output{cipher block}
\BlankLine
 \caption{HCA algorithm - cipher}
     \SetAlgoLined
     \For{$i \leftarrow 1$ $\KwTo$ $128$}{
        $\phi_{m} \leftarrow createRule_{m}(K)$\;
        $\phi_{b} \leftarrow createRule_{b}(K)$\;
        $s \leftarrow \lambda(s, \phi_{m}, \phi_{b})$\;
        $s \leftarrow Shift_8(s, opposite(toggle\_direction(K[256])))$\;
        $K \leftarrow Shift_1(K, left)$\;
     }
\end{algorithm}

\begin{algorithm}[!ht]
\label{alg:decipher}
\SetKwInOut{Input}{input}
\SetKwInOut{Output}{output}
\Input{A key $K$ and cipher block $s$}
\Output{plaintext block}
\BlankLine
 \caption{HCA algorithm - decipher}
     \SetAlgoLined
     $K \leftarrow Shift_{127}(K, left)$\;
     \For{$i \leftarrow 1$ $\KwTo$ $128$}{
        $s \leftarrow Shift_8(s, toggle\_direction(K[256])$\;
        $\phi_{m} \leftarrow createRule_{m}(K)$\;
        $\phi_{b} \leftarrow createRule_{b}(K)$\;
        $s \leftarrow \Phi(s, \phi_{m}, \phi_{b})$\;
        $K \leftarrow Shift_1(K, right)$\;
     }
\end{algorithm}

\section{Reversibility of HCA}
\label{sec:reversibility}

As described in Section \ref{sec:proposedsolution}, for each plaintext block, the HCA symmetric cryptography method employs a series of pre-image calculus operations ($ \lambda $) to cipher it, and a series of CA forward evolution operations ($ \phi $) is used to reverse the resulting cyphertext back to the original plaintext. 

Such mechanism can only be effective if the HCA model provides CA reversibility, so that any lattice configuration will only have a single pre-image. And thus, a required formal analysis of the reversibility property for HCA is provided in this section. 

\begin{theorem}
\label{theorem:reversibility}
Let $\phi_m$ and $\phi_b$ be, respectively, the main and boundary rules used at any step of
the HCA model. For any configuration $s^{t}$, $s^{t}$ has one and only one pre-image, which is $s^{t-1}$.
\end{theorem}

\begin{proof}
Let $s^{t}[0, \dots, 7]$ be the cells in $s^{t}$ under $\phi_b$, and $s^{t}[8, \dots, 127]$ the cells in $s^{t}$ under $\phi_m$.
For simplicity, we consider only that case of rule application over the lattice cells.
This is done without loss of generality due to the toroidal arrangement (wrap-around) of the CA lattice: any case can be shifted to fit this description.
Let $s^{t-1}[0, \dots, 127]$ be the cells in $s^{t-1}$.
Also w.l.o.g., let us consider only left-toggle rules.

By definition, we have that left-toggle rules necessarily change their output
when the left-most input bit of the neighborhood is changed, if the others bits remain unchanged. From the definition of the HCA model, we know that $\phi_m$ is left-toggle, and that $\phi_b$ is even stricter: only the left-most bit matters when computing the output, that is, the rule either always copies or always inverts the left-most input bit. Thus, even though $s^{t}[0]$ is the result of the application of $\phi_b$ over the neighborhood $s^{t-1}[124, \dots, 4]$, it depends solely on $s^{t-1}[124]$, as shown in Figure~\ref{fig:rev_1}.

\begin{figure}[!ht]
  \begin{center}
    \includegraphics[width=0.75\columnwidth]{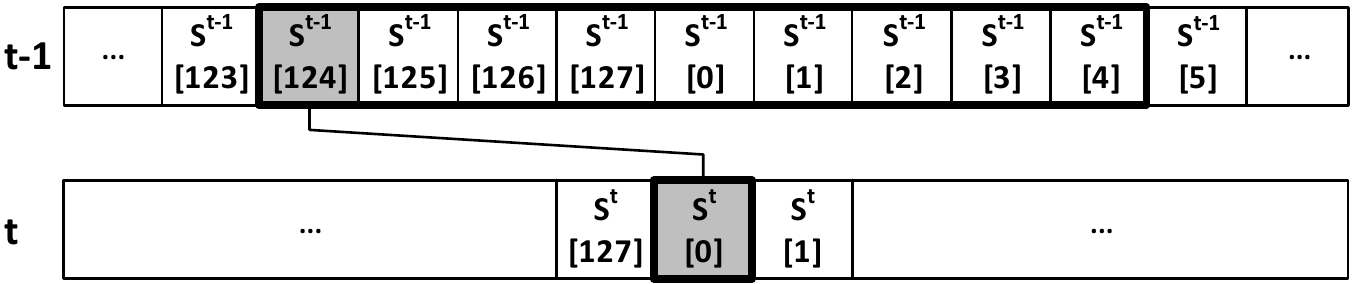}
  \end{center}
  \caption{Pre-image calculus on single bit of the border.}
  \label{fig:rev_1}
\end{figure}

Therefore, $s^{t}[0] = \phi_{b}(s^{t-1}[124])$, but also $s^{t-1}[124] = \phi_{b}(s^{t}[0])$, since $\phi_{b}$ can only express a copy or complement operation. The same is true for each lattice cell in $s^{t}[0, \dots, 7]$ with respect to each lattice cell in $s^{t-1}[124, \dots, 3]$. Therefore, we can apply $\phi_b$ over each element in $s^{t}[0, \dots, 7]$ to uniquely determine $s^{t-1}[124, \dots, 3]$, as shown in Figure~\ref{fig:rev_2}.

\begin{figure}[!ht]
  \begin{center}
    \includegraphics[width=0.9\columnwidth]{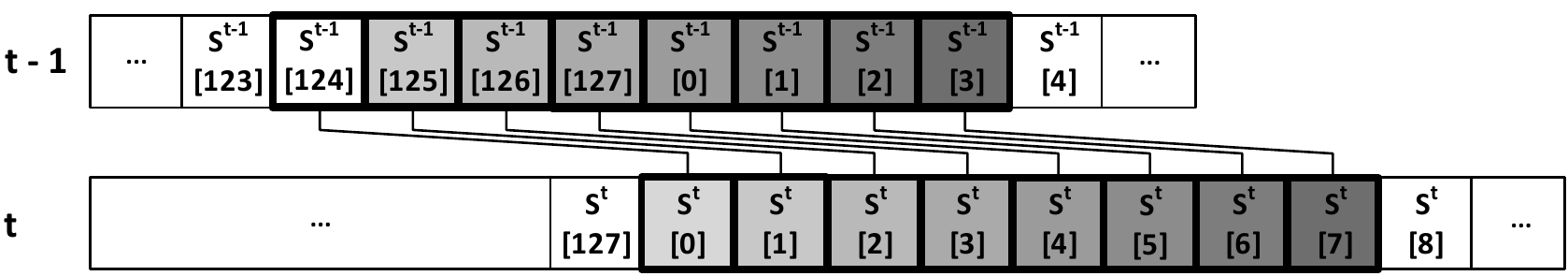}
  \end{center}
  \caption{Pre-image calculus on all bits of the border.}
  \label{fig:rev_2}
\end{figure}

Cell $s^{t}[127]$ is the output of rule $\phi_{m}$ on $s^{t-1}[123, \dots, 3]$. Since $s^{t-1}[124, \dots, 3]$ are already known, $s^{t-1}[123]$ can be computed by checking which bit value placed in $s^{t-1}[123]$ would lead to $\phi_{m}(s^{t-1}[123, \dots, 3]) = s^{t}[127]$, pictured in Figure~\ref{fig:rev_3}.

\begin{figure}[!ht]
  \begin{center}
    \includegraphics[width=0.8\columnwidth]{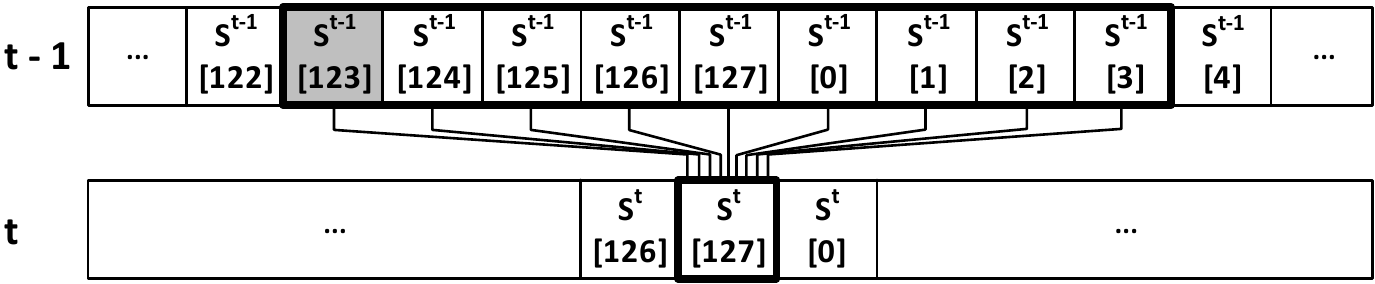}
  \end{center}
  \caption{Pre-image calculus in a main bit.}
  \label{fig:rev_3}
\end{figure}

Since $\phi_{m}$ is a left-toggle rule, any value change in $s^{t-1}[123]$ would result in a change to $s^{t}[127]$, therefore $s^{t-1}[123]$ is unique. Cells $s^{t-1}[122, 121, \dots, 5, 4]$ are sequentially computed in an analogous manner and therefore are also unique.

Since the pre-image $s^{t-1}$ is computed deterministically and uniquely, we have that the theorem holds.

The proof is analogous for the right-toggle rules case, one of the main changes in said case is that main bits are computed from the left to the right, again without loss of generality due to the wrapping property.
\end{proof}

From Theorem~\ref{theorem:reversibility} it follows that the HCA model is reversible.

\section{Formal Analysis}
\label{sec:formal}

An investigation developed in this article was the transformation of HCA calculus into a graph of a deterministic finite automata DFA with output, to investigate the safety potential of the HCA seeking to relate it to a theoretical model. The forward temporal evolution of a standard CA (homogeneous, synchronous and periodic contour) can be transformed into a DFA with output, regardless of the applied rule, being possible to model it as a Moore Machine or Mealy Machine [\cite{sutner:91}]. The Moore Machine modelling for HCA has a finite set of $Q$ states, an initial state $s_0$, an input alphabet $\Sigma_1 = \{0,1\}$, an output alphabet $\Sigma_2 = \{0,1,\varepsilon\}$, a set of transitions $\delta$: $Q$ x $\Sigma_1 \rightarrow Q$ is a function $\lambda: Q  \rightarrow \Sigma_2$ that defines the output associated with each state. In this model, the first bits have empty output ($\varepsilon$) and corresponds to the $m-1$ bits. From the $ m$-th state, we have the contour rule (state $c_i$), e.g. rule 15: \{$11110000$\}. The output associated with the state of the automaton is related to the output of transition rule (state $q_i$), the toggle rule, e.g. \{$01111000$\} (rule 30). Thus, the graph shown in Figure \ref{fig:mooreMachine} (left) models the HCA cipher. The HCA backward step can also be modeled by a Moore Machine (see Figure \ref{fig:mooreMachine} (right)). The topology of the graph varies according to the CA transition rule.

\begin{figure}[!ht]
  \begin{center}
    \includegraphics[width=\columnwidth]{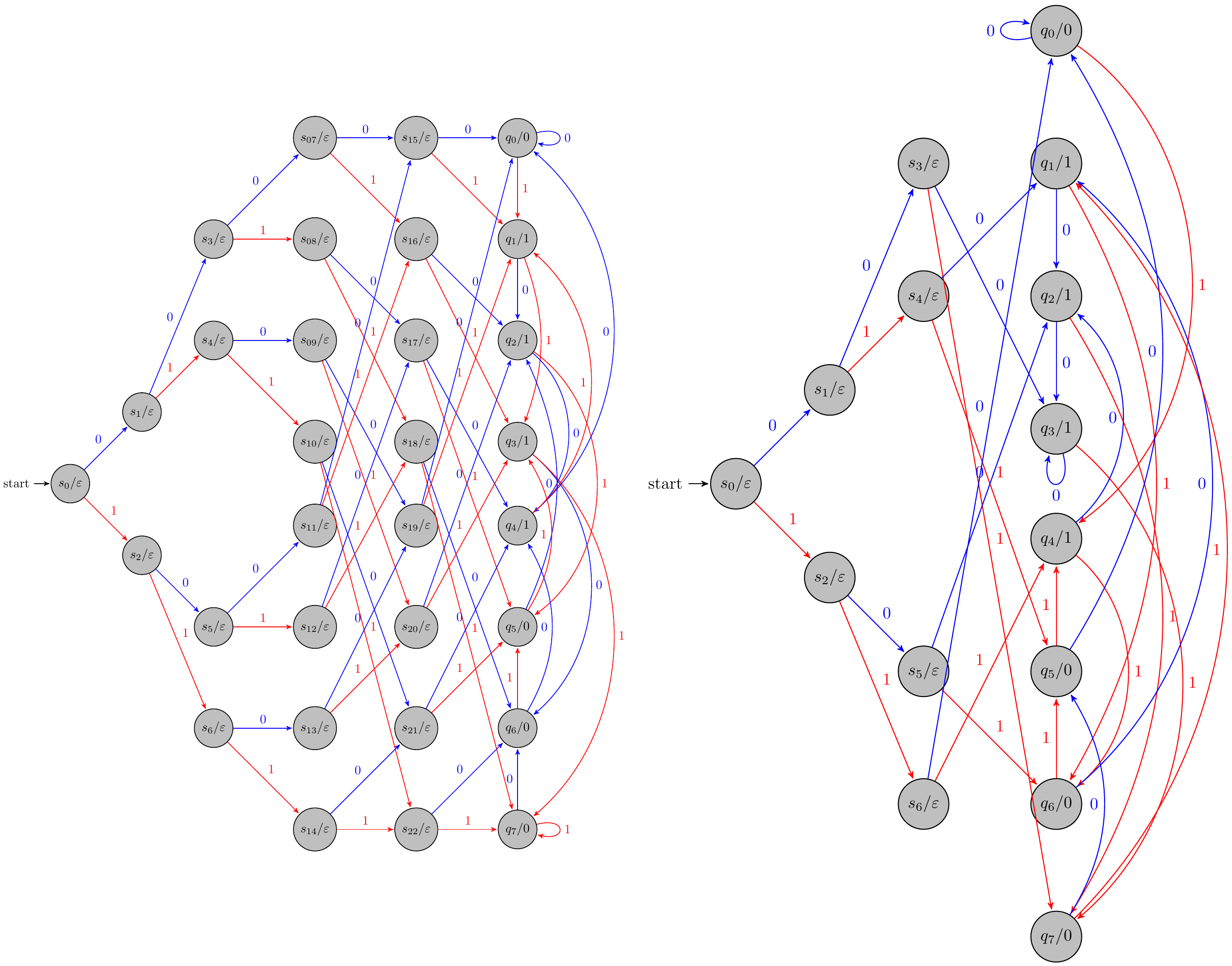}
  \end{center}
  \caption{Moore machine. Left: forward execution (cipher). Right: backward (decryption).}
  \label{fig:mooreMachine}
\end{figure}

Consider that the left-toggle rule ($r=1$) and suppose $\{000 \rightarrow b_0$, $001 \rightarrow b_1$, $010 \rightarrow b_2$, $011 \rightarrow b_3$, $100 \rightarrow b_4$, $101 \rightarrow b_5$, $110 \rightarrow b_6$, $111 \rightarrow b_7\}$. We consider $b_0 = \overline{b_3}$, $b_1 = \overline{b_5}$, $b_2 = \overline{b_6}$ e $b_3 = \overline{b_7}$. The calculus from an initial lattice ($[I_1, I_2, I_3, I_4, I_5, I_6, I_7, I_8]$ corresponds to $[1, 0, 1, 1, 0, 1, 1, 0]$). The pre-image ($[P_1, P_2, P_3, P_4, P_5, P_6, P_7, P_8, P_9, P_{10}]$ corresponds to $[\_, \_, \_, \_, \_, \_, \_, ?, 0, 1]$). The cells $P_9 = 0$ and $P_{10} = 1$ are known. This rule corresponds to $\{?00 \rightarrow 0$, $?00 \rightarrow 1$, $?01 \rightarrow 1$, $?01 \rightarrow 1$, $?10 \rightarrow 1$, $?10 \rightarrow 0$, $?11 \rightarrow 0$, $?11 \rightarrow 0\}$. Considering that the partial neighborhood are $P_9$ e $P_{10}$ and output bit is $I_8$, we have the triplet ($P_9P_{10}I_8 \rightarrow P_8$). If the output bit sequence of the left-toggle rule is $b_0b_1b_2b_3b_4b_5b_6b_7$ it is the same rule previously defined. Thus, if the direct rule is rule 30: $\{000 \rightarrow 1$, $001 \rightarrow 0$, $010 \rightarrow 0$, $011 \rightarrow 1$, $100 \rightarrow 1$, $101 \rightarrow 0$, $110 \rightarrow 1$, $111 \rightarrow 0\}$. 

Therefore, given the direct rule, the same rule can be used as the inverse. 
To exemplify a backward step, the toggle rule is: 30 \{01111000\} and the contour rule is: 15 \{11110000\}. The initial state is $ s_0 $ and the input is $\{0100101\}$. The first symbol to be read is ``1'' which is positioned below $ P_2$. Thus, the machine goes to state $ s_2 $ and returns 0 as output (contour rule). The process repeats until the entire tape is read. The machine will scroll (see Figure \ref{fig:mooreMachine} (right)): ($s_0, s_2, s_5, q_6, q_1, q_6, q_1, q_2$) and will output the sequence: \{0110101\} representing the pre-image relative to the initial lattice. 

Although the graphs obtained by the HCA model (forward and backward steps) are different from the CA graphs [\cite{sutner:91}], which does not present the contour rule in its topology, the conclusion of our study is that we must concentrate in the cyclic portion of the graphs, represented by the nodes $q_i$ and their transitions. In this case, there is no distinction between the graph of Figure \ref{fig:mooreMachine} (right) and the graph of a homogeneous AC, since both are based on rule 30. The safety of the method must be analyzed in relation to this cyclic part, which represents the main rule processing in the HCA method.

Cryptography inherently requires average-case intractability, it means that problems for which random instances have a very hard solution [\cite{peikert2016decade}]. This is substantially different from the notion of hardness usually considered in algorithm theory and NP-completeness, where a problem is considered difficult if there are only a few intractable instances. There are many problems that are hard in the worst case but are easier on the average. Especially for distributions that produce instances having some extra structure, e.g., the existence of a secret key for decryption. In HCA we avoided the structured periodical, fixed-point or null rules, this is the main reason why we only use a set $K$ of rules with entropy greater than $s> 0.75$, that represents chaotic rules or complexity rules (in the edge of chaos rules). In works [\cite{ajtai1996generating, ajtai1998shortest, ajtai1999generating, ajtai2007first}] the authors gave a connection between the worst case and the average case for problems instances. The authors proved that certain problems are hard on the average. On the other hand, other algorithms are hard only in worst cases. Using results of this kind, it is possible to design cryptography constructions and prove that they are infeasible to break, except when all instances of certain problems are simple to solve [\cite{peikert2016decade}]. Another public-key cipher scheme was proposed in [\cite{lin1995new}], using keys that can be easily generated. For security analysis, the authors examined some possible attacks including the integer knapsack problem is reduced to the linear Diophantine equation problem with the reduction process. The RSA [\cite{rivest1978method}] is an asymmetric key cryptography algorithm based on theories of numbers, which its safety is based on the difficulty of factorize large numbers.

Herein, the purpose is to analyze the encryption algorithm HCA in an attempt to associate it with an NP-Complete problem. As it was possible to model the main HCA step (pre-image calculation), as a graph, we believe that problems studied in Graph Theory, help us to find this association with a NP problem. A property that has already been found is the fact that the graph of the sensible rules has a Hamiltonian circuit, which is observed that all graphs have 2 Hamiltonian circuits (using $q_i$ states). As an example, in the graph of Figure \ref{fig:mooreMachine} (right), associated with rule 30, there are 2 Hamiltonian circuits: $(q_0, q_4, q_6, q_1, q_2, q_3, q_7, q_5, q_0)$ and $(q_0, q_4, q_2, q_3, q_7, q_1, q_6, q_5, q_0)$. In addition to the verification of this property, all rule output bits would be recorded on the tape not necessarily in the order of the rule. Thus, by analyzing a single pre-image calculation step, if we knew the state's Hamiltonian circuit, it would be possible to discover the bits of the key. But reordering these bits to find out the correct key is a permutation operation between the bits, which makes this problem of order $O(n!)$, considered a NP-problem. Find the Hamiltonian circuit in a graph is a NP-complete problem. Additionally, all the pre-HCA graphs share a very particular structure: the vertex degree is even (that could be a trivial graph). We believe that the existence of the Hamiltonian circuit in the pre-image graph could be one evidence of the safety of the method, but does not prove the safety by itself. In this way, other forms of method safety analysis were evaluated herein to prove their shuffling and safety power.

\section{Security}
\label{sec:security}

In the literature there are many methods that help evaluating how secure a cryptographic algorithm really is, and some that apply to symmetric cryptography will be listed and explained in this section.

\subsection{Ciphertext Information Entropy}
\label{sec:security_ciphertextEntropy}

In 1948, Claude E. Shannon, also known as the father of information theory, introduced the concept of Information Entropy [\cite{shannon1948}]. This concept can be seen as a way to measure the diversity of a certain event in a series of events. 

A normalized Expression (\ref{eq_1}) was presented in Section \ref{sec:proposedsolution} and used to select valid cryptographic keys for HCA. When using this expression to evaluate diversity in data, a higher result closer to $1.0$ indicates more diversity which is desirable in a cryptographic context. Information entropy is used to measure cryptographic strength in papers such as [\cite{ahmad2009, sun2010, blackledge2013}].

\subsection{Avalanche Effect}
\label{sec:security_avalanche}

Initially coined by Horst Feistel [\cite{feistel1973}], the ``Avalanche Effect'' is an expected propriety in cryptographic systems which can be measured in two cases [\cite{gustafson1994}]:
\begin{itemize}
\item \textbf{Plaintext Avalanche}: Using the same key, what is the impact of flipping a single bit in the plaintext?
\item \textbf{Key Avalanche}: Using the same plaintext, what is the impact of flipping a single bit in the key?
\end{itemize}
This procedure is detailed in Figure \ref{fig-avalancheDetailed}. When the Plaintext Avalanche is being evaluated we have $K = K'$, otherwise when Key Avalanche is being measured we have $X = X'$. For both cases, $Z = Y \oplus Y'$.

\begin{figure}[!ht]
  \begin{center}
    \includegraphics[width=0.75\columnwidth]{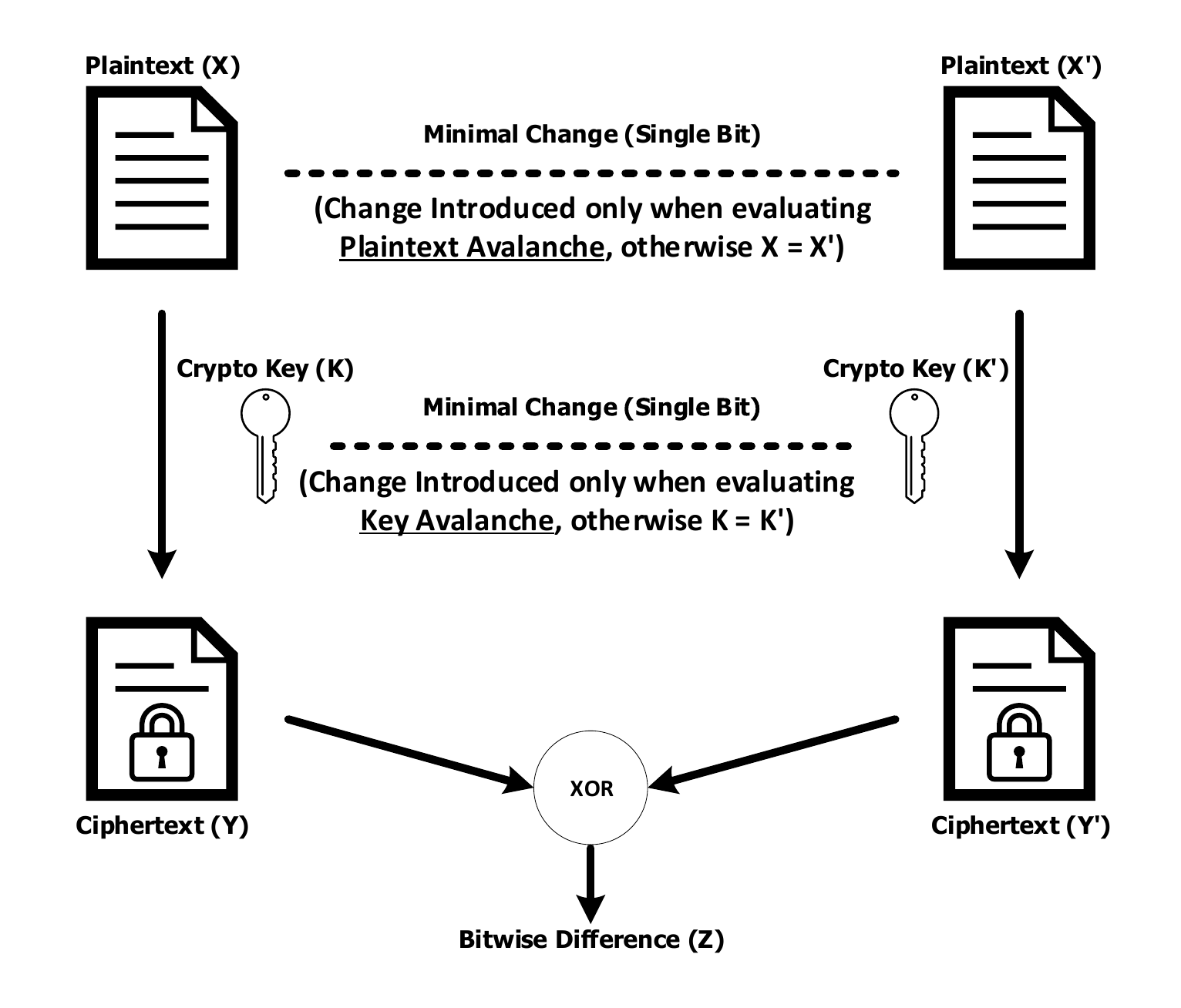}
  \end{center}
  \caption{Avalanche Effect evaluation explanation.}
  \label{fig-avalancheDetailed}
\end{figure}

If an algorithm does not exhibit sufficient Avalanche Effect compliance it would be extremely vulnerable to chosen-plaintext attacks, so this kind of analysis is regarded as a conventional test ran to evaluate cryptographic algorithms' strength [\cite{ramanujam2011, mishra2011, nadu2018}]. The method presented in Section \ref{sec:proposedsolution} was tested according to the specifications listed in \ref{sec:security_avalancheStdDeviation} and \ref{sec:security_avalancheEntropy}. Results are presented on Section \ref{sec:experiments_avalancheEffect}.

\subsubsection{Avalanche Effect - Standard Deviation Analysis}
\label{sec:security_avalancheStdDeviation}

If the algorithm presents strong Avalanche Effect, then we should expect the minimal difference between X and X' (for the Plaintext Avalanche test) or between K and K' (for the Key Avalanche test) to cause a significant difference between ciphertexts Y and Y' and thus, in ideal conditions, the percentage of `1' bits in Z should be around 50\% (as would also be expected from a randomly generated binary sequence). It would also be relevant to know, considering many distinct avalanche evaluations in a diverse population, how consistent are these results. In this case, the standard deviation analysis is considered a proper way to measure this, and a lower StdDev value would be desirable.

\subsubsection{Avalanche Effect - Entropy Analysis}
\label{sec:security_avalancheEntropy}

When evaluating the Avalanche Effect on an algorithm, besides counting how many bits of the ciphertext were affected by a minimal change, its also important to quantify how well propagated were the effects of said change. It would be desirable for this impact (bits changed) to be strongly distributed through the entire resulting ciphertext, and the concept of Information Entropy presented in \ref{sec:security_ciphertextEntropy} is a means to evaluate this.

So, when using the normalized entropy formula \ref{eq_1} to analyze sets of Z strings obtained from many Avalanche experiments, resulting values closer to $1.00$ are desirable, as they would indicate the propagated changes were highly dispersed across the resulting ciphertext. Meanwhile, a result close to $0.00$ is highly undesirable since it means the initial change made small or no difference in the ciphertext, or that it caused all the bits in Y and Y' to be the exact opposites.

\subsection{Birthday Attack}
\label{sec:security_birthdayAttack}

The Birthday Attack is a standard cryptanalytic technique in which reduction functions, such as hash operations, are analyzed for possible vulnerabilities based on the likelihood of collisions. Since there are no evident reduction functions in HCA, this kind of analysis does not apply to the algorithm.

\subsection{Meier-Staffelbach Attack}
\label{sec:security_meierAttack}

In [\cite{wolfram1986}], Stephen Wolfram proposed that CA rule 30 could be used as basis for a good PRNG (Pseudo-Random Number Generator) and thus a encryption mechanism could potentially be devised from it. But this proposal was further investigated by other authors such as Willi Meier and Othmar Stdelbach in [\cite{meier1991}], who found a vulnerability in the perceived randomness of rule 30 which we call the `Meier-Staffelbach Attack'.

This vulnerability was related to the specific design of rule 30 and, as other rules showed distinct behaviors, some of them were used in other similar cryptographic algorithms proposed since then, such as [\cite{nandi1994}] and [\cite{tomassini2000}].

The CA rules used in our algorithm are dynamically chosen according to the provided cryptographic key and are also switched at each step due to the key circular shift mechanism described in Section \ref{sec:proposedsolution}. So, even if a rule with a weakness similar to rule 30 was to be used in part of the process, it would not be the only applied rule. These characteristics allow us to assert our algorithm is not vulnerable to this kind of attack.

\subsection{Linear Cryptanalysis}
\label{sec:security_linearCryptanalysis}

In [\cite{matsui1993}], Mitsuru Matsui proposed `Linear Cryptanalysis', a known-plaintext attack in which the attacker tries to find a linear expression that embodies the differences between a plaintext and its resulting ciphertext in order to understand the cryptographic algorithm and then exploit its vulnerabilities.

The binary transformations applied during the encryption process are directly related to which CA rules are being used to evolve the plaintext, and since many of these CA rules provide nonlinear transformations, it would not be possible to find simple linear expressions that conveys the encryption mechanism. And thus linear cryptanalisis would not be a viable attack against HCA.

\subsection{Diferential Cryptanalysis}
\label{sec:security_diferentialCryptanalysis}

The `Diferential Cryptanalysis' attack was proposed in [\cite{biham1991}] to exploit vulnerabilities in the DES (Data Encryption Standard) algorithm [\cite{fips1999}]. This attack is based on studying how changes in the plaintext affect the resulting ciphertext.

If small changes in the plaintext cause limited effects on the ciphertext, this could be exploited as a vulnerability. If an encryption algorithm displays a high level ``Avalanche Effect" (described in \ref{sec:security_avalanche}) it can be considered safe against the diferential cryptanalysis attack.

\subsection{NIST PRNG Statistical Test Suite}
\label{sec:security_nistSuite}

The NIST (National Institute of Standards and Technology) is an American Institute founded in 1901 that provides guidelines on security and innovation. In 2010, NIST released their latest version of a Statistical Test Suite that evaluates the statistical quality of sequences generated by pseudorandom number generators (PRNG).

Since there is a known correlation between PRNG and encryption, this test suite is also being used to measure the quality of encryption algorithms by evaluating the statistical difference between the plaintext and its resulting ciphertext.

The NIST suite consists of 15 tests, and each test can be comprised of many subtests, which is why the suite is sometimes listed as having 15 tests [\cite{lakra2018}] and in other times as being a set of 188 or more tests [\cite{manzoni2018}].

\section{Evaluations}
\label{sec:experiments}

\subsection{Avalanche Effect}
\label{sec:experiments_avalancheEffect}

The Avalanche Effect test was applied to HCA in the following conditions:
\begin{itemize}
\item The initial AC lattice is comprised of N bits (128, 256 or 512 bits)
\item For each N value, $N^{2}$ random initial lattices are generated
\item Each execution ran for N evolution steps
\end{itemize}

Each of the $N^{2}$ randomly generated initial lattices for each N value was encrypted by the HCA algorithm using random valid HCA cryptographic keys (with spatial entropy $> 0.75$) and the results are presented in Table \ref{table-hca_r4_avalanche_statistics}. In each table, values obtained from $N^{2}$ N-sized sequences generated by a generic PRNG are included for comparison purposes.

\begin{table}[!ht]
\centering
\caption{HCA - Avalanche Effect Result Statistics}
\label{table-hca_r4_avalanche_statistics}
\begin{tabular}{|c|c|c|c|c|c|c|}
\hline
{\textbf{N}} & \multicolumn{2}{c|}{\textbf{HCA - Text Aval.}} & \multicolumn{2}{c|}{\textbf{HCA - Key Aval.}} & \multicolumn{2}{c|}{\textbf{RNG}} \\ \cline{2-7} & \textbf{Avg (\%)} & \textbf{$\sigma$} & \textbf{Avg (\%)} & \textbf{$\sigma$} & \textbf{Avg (\%)} & \textbf{$\sigma$} \\ \hline
\textbf{128} & 49.937 & 4.399 & 50.064 & 4.422 & 50.034 & 4.464 \\ \hline
\textbf{256} & 50.036 & 3.113 & 49.940 & 3.125 & 49.999 & 3.133 \\ \hline
\textbf{512} & 49.994 & 2.199 & 50.016 & 2.217 & 50.002 & 2.202 \\ \hline
\textbf{AVG} & 49.989 & 3.237 & 50.007 & 3.255 & 50.012 & 3.266 \\ \hline
\end{tabular}
\end{table}

In Table \ref{table-hca_r4_avalanche_statistics} the test results are displayed for the plaintext and key avalanche evaluations. A good encryption algorithm should present a modification rate in the final ciphertext around 50\%, with a low standard deviation ($\sigma$), as is the case in all average values found for both instances. The result values are also similar to the bit distribution rate in the PRNG generated sequences.

Its also important to ensure a random spatial dispersion of the changed bits, so an spatial entropy analysis is done on the resulting difference (XOR) lattice. These results follow in Tables \ref{table-hca_r4_plaintext_entropy} and \ref{table-hca_r4_key_entropy} for the plaintext and key avalanche evaluations, respectively.

\begin{table}[!ht]
\centering
\caption{HCA - Plaintext Avalanche Entropy Analysis}
\label{table-hca_r4_plaintext_entropy}
\begin{tabular}{|c|c|c|c|c|c|c|c|c|}
\hline
{\textbf{N}} & \multicolumn{4}{c|}{\textbf{HCA}} & \multicolumn{4}{c|}{\textbf{RNG}} \\ \cline{2-9} & \textbf{Min} & \textbf{Max} & \textbf{Avg} & \textbf{$\sigma$} & \textbf{Min} & \textbf{Max} & \textbf{Avg} & \textbf{$\sigma$} \\ \hline
\textbf{128} & 0.783 & 0.940 & 0.883 & 0.018 & 0.794 & 0.938 & 0.883 & 0.018 \\ \hline
\textbf{256} & 0.849 & 0.936 & 0.897 & 0.011 & 0.843 & 0.940 & 0.897 & 0.011 \\ \hline
\textbf{512} & 0.877 & 0.931 & 0.908 & 0.007 & 0.874 & 0.933 & 0.908 & 0.007 \\ \hline
\textbf{AVG} & 0.836 & 0.936 & 0.896 & 0.012 & 0.837 & 0.937 & 0.896 & 0.012 \\ \hline
\end{tabular}
\end{table}

\begin{table}[!ht]
\centering
\caption{HCA - Key Avalanche Entropy Analysis}
\label{table-hca_r4_key_entropy}
\begin{tabular}{|c|c|c|c|c|c|c|c|c|}
\hline
{\textbf{N}} & \multicolumn{4}{c|}{\textbf{HCA}} & \multicolumn{4}{c|}{\textbf{RNG}} \\ \cline{2-9} & \textbf{Min} & \textbf{Max} & \textbf{Avg} & \textbf{$\sigma$} & \textbf{Min} & \textbf{Max} & \textbf{Avg} & \textbf{$\sigma$} \\ \hline
\textbf{128} & 0.797 & 0.939 & 0.883 & 0.018 & 0.794 & 0.938 & 0.883 & 0.018 \\ \hline
\textbf{256} & 0.845 & 0.934 & 0.897 & 0.011 & 0.843 & 0.940 & 0.897 & 0.011 \\ \hline
\textbf{512} & 0.879 & 0.932 & 0.908 & 0.007 & 0.874 & 0.933 & 0.908 & 0.007 \\ \hline
\textbf{AVG} & 0.840 & 0.935 & 0.896 & 0.012 & 0.837 & 0.937 & 0.896 & 0.012 \\ \hline
\end{tabular}
\end{table}

The average entropy results for each tested radius value should be as close as possible to $1.00$ and the evaluated results, presented above, were comparable in all instances to the average entropy found in randomly generated sequences.

\subsection{NIST PRNG Suite}
\label{sec:experiments_nistSuite}

The NIST suite tests were run on sequences that directly convey the changes an encryption method causes on random plaintexts. Each test has a minimum input size recommendation, and since some of them are as large as $10^{6}$ bits, the results presented here were obtained using sequences consisting of 10 Megabytes. The input sequence construction procedure is explained on Figure \ref{fig-nistInput}.

\begin{figure}[!ht]
  \begin{center}
    \includegraphics[width=0.7\columnwidth]{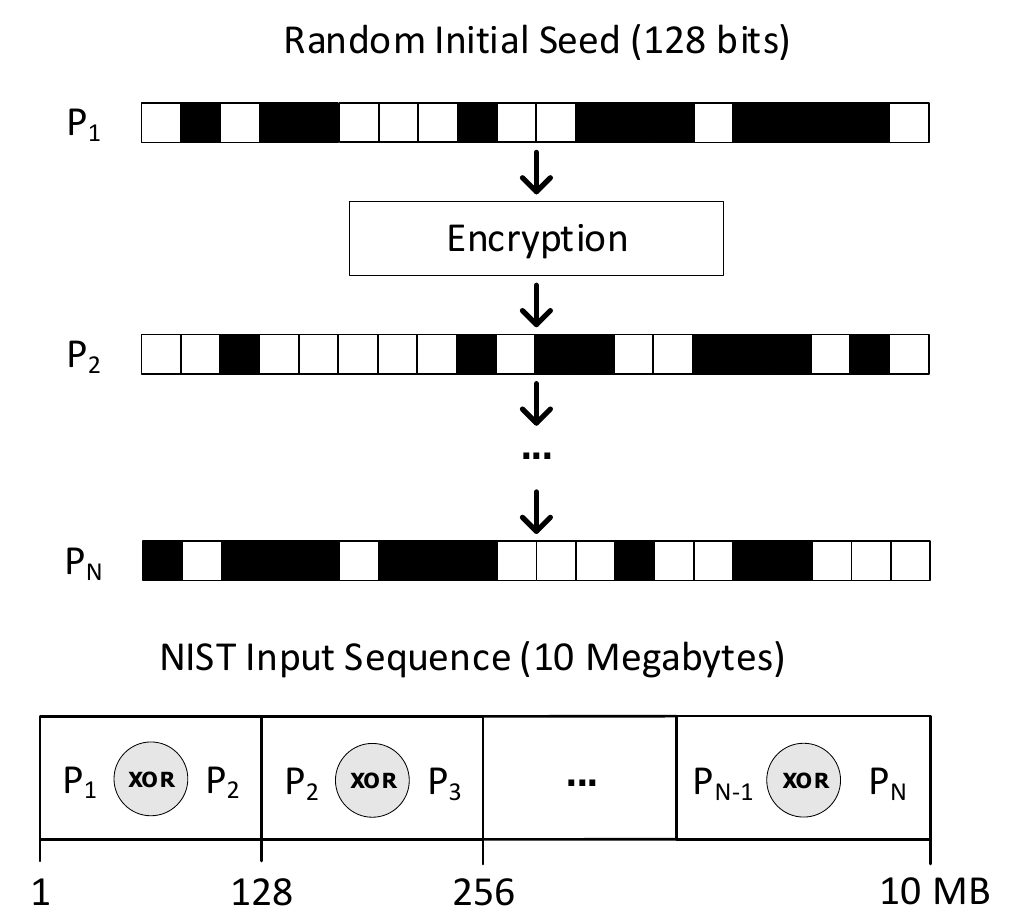}
  \end{center}
  \caption{NIST input sequence building.}
  \label{fig-nistInput}
\end{figure}

As presented in Figure \ref{fig-nistInput}, building each 10 Megabytes sequence begins by initializing a single 128-bit block sized lattice using a pseudo-random seed, this initial plaintext is regarded as ``$P_{1}$''. The encryption algorithm is applied to $P_{1}$, generating a ciphertext called ``$P_{2}$'', and their binary difference, $P_{1} \oplus P_{2}$, represents the effect of the encryption procedure. After $P_{1} \oplus P_{2}$ is calculated, this 128-bit sequence is the first part of the 10 megabytes input sequence used for NIST evaluation; the next part will be $P_{2} \oplus P_{3}$, where ``$P_{3}$'' is the new ciphertext obtained by running the encryption algorithm with $P_{2}$ as the plaintext. This iterative procedure is repeated until the 10 Megabytes sequence is complete by appending the last part, $P_{N-1} \oplus P_{N}$, where, accordingly, $N = (10\:megabytes)/(128\:bits)$.

The NIST evaluation ran for each algorithm was executed for 1,000 distinct 10 Megabytes sequences generated using the procedure listed above. The percentage of passing sequences for each NIST test follow in Table \ref{table-hca_nist}.

\begin{table}[!ht]
\centering
\caption{NIST Suite Tests}
\label{table-hca_nist}
\begin{tabular}{|l|c|c|}
\hline
\multicolumn{1}{|c|}{\textbf{NIST Test}} 					& \textbf{HCA} & \textbf{AES} \ \\ \hline
\textbf{T01 - Frequency (Monobits) Test} 					& 99.1\% 		& 99.1\% 		\\ \hline
\textbf{T02 - Frequency Test within a Block} 				& 99.1\% 		& 99.4\% 		\\ \hline
\textbf{T03 - Runs Test} 									& 99.0\% 		& 98.7\% 		\\ \hline
\textbf{T04 - Test for the Longest Run of Ones in a Block} 	& 98.5\% 		& 98.3\% 		\\ \hline
\textbf{T05 - Binary Matrix Rank Test} 						& 98.8\% 		& 99.2\% 		\\ \hline
\textbf{T06 - Discrete Fourier Transform (Specral) Test} 	& 98.2\% 		& 98.6\% 		\\ \hline
\textbf{T07 - Non-Overlapping Template Matching Test} 		& 97.4\% 		& 98.0\% 		\\ \hline
\textbf{T08 - Overlapping Template Matching Test} 			& 99.0\% 		& 99.2\% 		\\ \hline
\textbf{T09 - Maurer’s “Universal Statistical” Test} 		& 99.1\% 		& 99.0\% 		\\ \hline
\textbf{T10 - Linear Complexity Test} 						& 99.2\% 		& 98.4\% 		\\ \hline
\textbf{T11 - Serial Test} 									& 98.7\% 		& 98.1\% 		\\ \hline
\textbf{T12 - Approximate Entropy Test} 					& 99.3\% 		& 99.1\% 		\\ \hline
\textbf{T13 - Cumulative Sums (Cusum) Test} 				& 98.6\% 		& 98.8\% 		\\ \hline
\textbf{T14 - Random Excursions Test} 						& 93.1\% 		& 93.2\% 		\\ \hline
\textbf{T15 - Random Excursions Variant Test} 				& 93.1\% 		& 92.7\% 		\\ \hline
\end{tabular}
\end{table}

Besides the results found for sequences generated by the HCA method, Table \ref{table-hca_nist} also contains the results for sequences similarly generated using the AES algorithm. The proximity between the passing-rates of both algorithms, for all tests in the NIST suite, suggests HCA is a promising method, since AES is the current symmetric encryption standard.

\section{Conclusions and Future Work}
\label{sec:conclusion}

This paper describes a symmetrical block cipher cryptographic model based on reversible heterogeneous cellular automata that employs two radius-4 toggle-rules. The main rule is chaotic and non-additive; it is applied to the majority of bits at each time step to provide the necessary entropy to the encryption process. The second one is periodic (more specifically, fixed-point with an spatial displacement) and additive; it is applied to a small set of consecutive bits (the lattice border) and is used to ensure the existence of a pre-image. This model was named HCA (Hybrid Cellular Automata) and it was firstly proposed in 2007, when a patent registration was submitted in Brazil [\cite{oliveira2007sistema}]. It is the first time that HCA is presented and evaluated in a wide-range scientific forum. In the past, just the brazilian patent registration (whose process was finalized in 2019) and some local academic works (master's thesis), written in portuguese, have focused on aspects of, and extensions to, the HCA model [\cite{magalhaes2010metodo, lima2012modelo, alt2013propriedades}].

The adopted block size is 128 bits and the secret key has 257 bits, where 256 of them define the main rule to be applied. Moreover, as presented here, forward and backward CA evolution procedures correspond to the decryption and encryption processes, respectively. However, the converse is also possible; HCA enables one to use forward evolution in ciphering, while the receiver must use backward evolution to decipher. In general, forward is faster than backward evolution and in the specification discussed here the receiver will employ the faster process to decipher. It would also be simple to increase the size of the block for 256 bits or more. If one wants to use a larger key space, it is also easy to adapt the model to use radius-5 toggle rules or more; however this would increase the complexity of implementing the solution in HPC systems, such as FPGAs [\cite{halbach2004implementing}].

When compared to other similar CA-based methods [\cite{gutowitz:95, oliveira:04, oliveira:08, oliveira:10, oliveira2010_ppsn, oliveira2011deeper, wuensche2008}] which also apply toggle rules, the reversible model used in the HCA algorithm has the advantage of keeping the ciphertext size equal to the plaintext, whereas being valid for any possible CA initial configuration since the existence of a pre-image is ensured [\cite{oliveira:08}]. As a symmetric algorithm, this model can be applied to any kind of data (text, images, etc.) by defining a safe padding strategy and a secure mode of operation.

The experimental results provide herein evidence that HCA is a robust cryptographic algorithm. This poses a strong argument in favor of further investigating a cryptographic algorithm based on cellular automata due to the inherent parallelism of the model that can be harnessed in proper hardware, in opposition to conventional algorithms such as AES that are mostly serial in nature. Additionally, two theoretical analyses were also presented. The first proves the reversibility of the CA model due to the heterogeneous arrange of the two toggle rules (chaotic and additive). The second one uses Graph Theory to show that the problem of breaking the secret key in HCA could be approximately reduced to the Hamiltonian Cycle Problem (HCP), which is known to belongs to NP-complete class. Despite the fact that HCP general formulation in Graph Theory was proposed for an arbitrary graph, while HCA defines a graph with a specific topology, this analysis points to the robust security of the cryptographic algorithm. 

Despite the parallelism mechanisms of HCA having already been explained in this paper, implementation in specialized hardware was not possible at the time of this publication. The method was implemented in conventional x86 software and only inter-block parallelism, which is available to any block-cipher algorithm, was explored. Therefore, the efficient implementation of HCA in High Performance Computing (HPC) systems, such as FPGA architectures, is an on-going work of our research group. From the theoretical point of view, the estimated time to perform the sequential calculation of HCA encryption (backward) or decryption (forward), for one block of bits, is $\phi_c = m \times N \times T$, where, $m =  2r +1$, $N$ is the lattice size and $T$ is the number of preimage steps. In the specification discussed here [\cite{oliveira2007sistema}], N = T = 128. On the other hand, the estimated time to perform the parallel calculation of HCA decryption for the same block of bits is given by $\phi_d =  (N - m) \times T$ and to perform the parallel HCA encryption is given by $\phi_c =  2T + N - m$. To achieve this theoretical gain, one must use $N$ processing nodes in the architecture. This estimation  ignores the effects of memory access and communication and other practical questions related to parallel implementation. However, it highlights the huge potential of this CA-based model considering HPC, specially FPGAs. 

The conception of the reversibility analysis presented in Section \ref{sec:reversibility} gave insight into the possibility of extending this concept to HCA alternatives with an even higher heterogeneity level. This could also allow the use of rules with radius lower than 4, which would lead to more performance and ease to implement. Another expected development is a forthcoming work that investigates an HCA adaptation using multidimensional cellular automata.

\section*{Declarations}

%

%
\subsection*{Conflict of interest}
The authors declare that they have no conflict of interest.
\subsection*{Availability of data and material}
Not applicable
\subsection*{Code availability}
An implementation of the featured algorithm is made available at https://github.com/evertonrlira/HCA. Any updates will also be published on the linked repository.
\subsection*{Ethics approval}
Not applicable
\subsection*{Consent to participate}
Not applicable
\subsection*{Consent for publication}
The authors declare that they give consent for publication.

\bibliographystyle{spbasic}      
\bibliography{main}   

\end{document}